\pdfoutput =1

\documentclass{article}

\usepackage[T1]{fontenc}
\usepackage[utf8]{inputenc}
\usepackage{amsmath,amsfonts,amsthm,amssymb}
\usepackage{subeqnarray}
\usepackage{setspace}
\usepackage{graphics,graphicx,color}
\usepackage{url}
\usepackage{enumerate}
\usepackage{float} 
\usepackage{multirow}
\usepackage{hhline}
\usepackage[margin=1.2in]{geometry}
\usepackage{braket}
\usepackage{dsfont} 
\usepackage{authblk}
\usepackage[makeroom]{cancel}

\newtheorem{theorem}{Theorem}
\newtheorem{cor}{Corollary}
\newtheorem{definition}{Definition}
\newtheorem{lem}{Lemma}

\newtheorem{claim}{Claim}
\usepackage[scr=boondoxo,scrscaled=1.05]{mathalfa}

\newcommand{\sandwich}[3]{\mbox{$ \Braket{ #1 | #2 | #3 } $}}

\newcommand{\average}[1]{\mbox{$\Braket{#1}$}}
\newcommand{\Tr}[1]{\mbox{$\mathrm{Tr}\left[#1\right]$}}
\newcommand{\pTr}[2]{\mbox{$\mathrm{Tr}_{#1}\left[#2\right]$}}
\newcommand{\Id}{\mathds{1}}
\newcommand{\beq}{\begin{eqnarray}}
\newcommand{\eeq}{\end{eqnarray}}

\newcommand\norm[1]{\left\Vert#1\right\Vert}

\newcommand{\ps}{\ket{\psi}}
\newcommand{\Ah}{\hat{A}}
\newcommand{\Bh}{\hat{B}}

\newcommand{\tr}{\mbox{Tr}\,}

\begin{document}
\title{Separation of finite and infinite-dimensional quantum correlations, with infinite question or answer sets}

\author{Andrea Coladangelo}
\author{Jalex Stark}

\affil{Computing and Mathematical Sciences,	
Caltech\\ 	
{\{acoladan,jalex\}@caltech.edu}
}

\maketitle

\abstract{Completely determining the relationship between quantum correlation sets is a long-standing open problem, known as Tsirelson's problem. Following recent progress by Slofstra \cite{Slofstra16,Slofstra17}, only two instances of the problem remain open. One of them is the question of whether the set of finite-dimensional quantum correlations is strictly contained in the set of infinite-dimensional ones (i.e. whether $\mathcal{C}_q \neq \mathcal{C}_{qs}$). The usual formulation of the question assumes finite question and answer sets. 
In this work, we show that, when one allows for either infinite answer sets (and finite question sets) or infinite question sets (and finite answer sets), there exist correlations that are achievable using an infinite-dimensional quantum strategy, but not a finite-dimensional one. For the former case, our proof exploits a recent result \cite{CGS16}, which shows self-testing of any pure bipartite entangled state of arbitrary local dimension $d$, using question sets of size $3$ and $4$ and answer sets of size $d$. For the latter case, a key step in our proof is to show a novel self-test, inspired by \cite{CGS16}, of all bipartite entangled states of any local dimension $d$, using question sets of size $O(d)$, and answer sets of size $4$ and $3$ respectively.

\section{Introduction}
Given question sets $\mathcal{X}$ and $\mathcal{Y}$ and answer sets $\mathcal{A}$ and $\mathcal{B}$, a (bipartite) \textit{correlation} is a  collection of conditional probability distributions $\{p(a,b|x,y): a\in \mathcal{A}, b \in \mathcal{B}\}_{(x,y)\in \mathcal{X}\times \mathcal{Y}}$. 
The long-standing problem of completely determining the relationship between variants of quantum correlation sets is known as Tsirelson's problem \cite{Tsi06,Fri12}. 

We let $\mathcal{C}_q$ be the set of correlations which can be realized by local projective measurements on a shared bipartite finite-dimensional quantum state in $\mathcal{H}_A\otimes \mathcal{H}_B$, for Hilbert spaces $\mathcal{H}_A$ and $\mathcal{H}_B$.  $\mathcal{C}_{qs}$ is the relaxation where we allow $\mathcal{H}_A$ and $ \mathcal{H}_B$ to be infinite-dimensional, while
$\mathcal{C}_{qa}$ is defined as the closure of $C_{q}$, i.e. limits of quantum correlations on finite-dimensional Hilbert spaces, and $C_{qc}$ is the set of possibly infinite-dimensional quantum correlations arising in the commuting operator model. These definitions implicitly assume that question and answer sets are finite.

Thanks to the containment  $\mathcal{C}_{qs} \subseteq \mathcal{C}_{qa}$, we know that $\mathcal{C}_{qa}$ is also the closure of $\mathcal{C}_{qs}$ \cite{Scholz08}. 
The following is the known hiearchy
\begin{equation}
\mathcal{C}_q \subseteq \mathcal{C}_{qs} \subseteq \mathcal{C}_{qa} \subseteq \mathcal{C}_{qc}
\end{equation}
with recent progress by Slofstra showing first that $\mathcal{C}_{qs} \neq \mathcal{C}_{qc}$ \cite{Slofstra16}, and later strengthening this to  $\mathcal{C}_{qs} \neq \mathcal{C}_{qa}$ \cite{Slofstra17}. 

The only two outstanding instances of Tsirelson's problem are whether $\mathcal{C}_{q} = \mathcal{C}_{qs}$ and whether $\mathcal{C}_{qa} = \mathcal{C}_{qc}$. In this work, we make progress related to the former. We show that if one considers either correlations on finite question sets and \textit{infinite} answer sets or correlations on \textit{infinite} question sets and finite answer sets, then there is separation between correlations arising from finite and infinite-dimensional quantum strategies.

We make these statements more precise. We let $\mathcal{C}_q^{m,n,r,s}$ and $\mathcal{C}_{qs}^{m,n,r,s}$ be the sets of quantum correlations that have question sets of size $m$ and $n$ and answer sets of size $r$ and $s$ respectively on finite and infinite-dimensional Hilbert spaces. We denote by $\mathcal{C}_{q}^{m,n,\infty,\infty}$ and $\mathcal{C}_{qs}^{m,n,\infty,\infty}$ their respective variants with answer sets of countably infinite size, and by $\mathcal{C}_{q}^{\infty, \infty, r,s}$ and $\mathcal{C}_{qs}^{\infty,\infty, r,s}$ their variants with question sets of countably infinite size. Since defining  $\mathcal{C}_{q}^{m,n,\infty,\infty}$ and $\mathcal{C}_{q}^{\infty, \infty, r,s}$ is not entirely unambiguous, we give a formal definition in Section \ref{sec: preliminaries}.

For later convenience, we denote by $\mathcal{C}_{q \leq N}^{m,n,r,s}$ the subset of $\mathcal{C}_{q}^{m,n,r,s}$ of correlations obtained by states of Schmidt rank at most $N$ (likewise when $m$ and $n$ or $r$ and $s$ are $\infty$). 

To the best of our knowledge, the first result giving a non-local game (with classical questions) whose optimal winning probability can be approximated arbitrarily well, but not achieved perfectly, with finite-dimensional quantum resources is found in  \cite{manvcinska2014unbounded}. The game has two questions per party and countably infinite answer sets. However, the sequence of correlations that the authors present does not have a limit, since they are uniform distributions on increasingly large sets. On the other hand, a candidate set of correlations which may be in $\mathcal{C}_{qs}$ and not in $\mathcal{C}_{q}$ is the set of correlations attaining maximal violation of the $I_{3322}$ Bell inequality \cite{I3322Froissart}. Here, numerical evidence suggests that finite-dimensional states can get arbitrarily close to the maximal violation, but are not enough to attain the maximum \cite{I3322Pal}. Unlike in the case of \cite{manvcinska2014unbounded}, where the limit of correlations does not exist, it is believed that correlations attaining maximal violation of $I_{3322}$ lie either in $\mathcal{C}_{qs}$ or in $\mathcal{C}_{qa}$; the two sets were recently shown to be different by Slofstra \cite{Slofstra17}. In \cite{Slofstra17}, Slofstra constructs a sequence of correlations in $\mathcal C_{qs}$ which has a limit, but not in $\mathcal C_{qs}$, proving that $C_{qs} \neq C_{qa}$, and hence that $C_{qs}$ is not closed. 

Our contribution is that we construct a sequence of correlations in $\mathcal C_q^{3,4, \infty,\infty}$ which has a limit, and we show that this limit is in $C_{qs}^{3,4,\infty, \infty}$ but not $C_q^{3,4,\infty, \infty}$, thus proving separation of $C_q^{3,4,\infty, \infty}$ and  $C_{qs}^{3,4,\infty, \infty}$. We also construct a sequence of correlations in $\mathcal C_q^{\infty,\infty,4,3}$ which has a limit that is in $C_{qs}^{\infty,\infty,4,3}$ but not in $C_q^{\infty,\infty,4,3}$, proving separation of $C_q^{\infty,\infty,4,3}$ and  $C_{qs}^{\infty,\infty,4,3}$. In both cases, we show that any finite amount of entanglement is not enough to achieve the limit, while one can write down a natural infinite-dimensional strategy that achieves the limit.

Our first main theorem is the following. 

\begin{theorem}
\label{thm1}
There exists a correlation $p^* \in \mathcal{C}_{qs}^{3,4,\infty,\infty}$ such that, if $p \in \mathcal{C}_{q \leq N}^{3,4,\infty,\infty}$ and $p$ is $\delta$-close to $p^*$ (according to the distance $|\cdot|_{corr}$ defined in Definition \ref{def: distance}), then  $N = \Omega\Big(\frac{1}{\delta^{1/32}}\Big)$.
\end{theorem}

Note that Theorem \ref{thm1} can be seen as a dimension witness. And since $p^* \in \mathcal{C}_{qs}^{3,4,\infty,\infty}$, as a corollary it immediately implies the separation:
\begin{cor}
\label{cor1}
$\mathcal{C}_{q}^{3,4,\infty,\infty} \neq \mathcal{C}_{qs}^{3,4,\infty,\infty} $
\end{cor}

Our second main theorem is the following.
\begin{theorem}
\label{thm2}
There exists a correlation $p^* \in \mathcal{C}_{qs}^{\infty,\infty, 4, 3}$ such that, if $p \in \mathcal{C}_{q \leq N}^{\infty,\infty, 4, 3}$ and $p$ is $\delta$-close to $p^*$ (according to the distance $|\cdot|_{corr}$ defined in Definition \ref{def: distance}), then  $N = \Omega\Big(\frac{1}{\delta^{1/32}}\Big)$.
\end{theorem}
Again, the theorem can be seen as a dimension witness, and it implies the separation:
\begin{cor}
\label{cor2}
$\mathcal{C}_{q}^{\infty,\infty,4,3} \neq \mathcal{C}_{qs}^{\infty,\infty,4,3} $
\end{cor}


Our proof of Theorem \ref{thm1}, covered in section \ref{infinite answers}, exploits a recent result of  \cite{CGS16}, which shows that any pure bipartite entangled state of qudits can be self-tested, using questions sets of size $3$ and $4$ and answer sets of size $d$. On the other hand, a key step in our proof of Theorem \ref{thm2}, covered in section \ref{infinite questions}, is to show a novel self-test for any bipartite entangled state of qudits, inspired by \cite{CGS16}, using question sets of size $O(d)$ and answer sets of size $4$ and $3$.

One can view our results as ``evidence'' that $\mathcal C_q\neq \mathcal C_{qs}$. On the other hand, one can find results giving evidence in favor of $\mathcal C_q = \mathcal C_{qs}$. For some classes of pseudotelepathy games, e.g.\ linear constraint games \cite{cleve2014characterization} and weak projection games \cite{manvcinska2014maximally}, we know that the ideal strategies must use maximally entangled states (which are inherently finite-dimensional). However, the methods used to prove results like these seem to rely heavily on the game structure. It is plausible that $\mathcal C_q \neq \mathcal C_{qs}$, but the separation is witnessed only by correlations which do not arise from non-local games with a binary (or integer-valued) scoring function.

\section{Preliminaries}
\label{sec: preliminaries}
\paragraph{Strategies, $\mathcal{C}_{q}^{m,n,\infty,\infty}$ and $\mathcal{C}_{q}^{\infty,\infty,r,s}$.} Let $\mathcal{X}, \mathcal{Y}$ be the questions sets, and $\mathcal{A}, \mathcal{B}$ the answer sets. In general, a strategy is specified by Hilbert spaces $\mathcal{H}_A$ and  $\mathcal{H}_B$, a pure state $\ket{\psi} \in \mathcal{H}_A \otimes \mathcal{H}_B$, and projective measurements $\{A^a_x\}_a$ on $\mathcal{H}_A$, $\{B^b_y\}_b$ on $\mathcal{H}_B$, for $x \in \mathcal{X}, y \in \mathcal{Y}$. For short, we refer to a strategy as a triple $\left(\ket{\psi}, \{A^a_x\}_a, \{B^b_y\}_b \right)$. Note that in order to concisely describe a strategy, we will sometimes simply specify the observables, which in turn determine the projective measurements. Note that a projective measurement can have countably infinite outcomes, and this simply means that it specifies a countably infinite set of eigenspaces, which of course requires the underlying Hilbert space to be infinite-dimensional. Nonetheless, we can still talk about finite-dimensional quantum correlations with countably infinite answer sets by adding the requirement that the joint state has finite Schmidt rank, even though the Hilbert space may be infinite-dimensional. This is how we define $\mathcal{C}_{q}^{m,n,\infty,\infty}$ and $\mathcal{C}_{q}^{\infty,\infty, r, s}$:

\begin{definition}($\mathcal{C}_{q}^{m,n,\infty,\infty}$ and $\mathcal{C}_{q}^{\infty,\infty,r,s}$)
Let $\mathcal{X}$, $\mathcal{Y}$ be question sets of size $m$ and $n$, and $\mathcal{A}$, $\mathcal{B}$ answer sets of countably infinite size. A correlation $\{p(a,b|x,y): (a,b) \in \mathcal{A} \times \mathcal{B}\}_{(x,y) \in \mathcal{X} \times \mathcal{Y}}$ is in $\mathcal{C}_{q}^{m,n,\infty,\infty}$ if there exist Hilbert spaces $\mathcal{H}_A$, $\mathcal{H}_B$, a strategy $\left( \ket{\psi}, \{A^a_x\}_a, \{B^b_y\}_b \right)$ on $\mathcal{H}_A \otimes \mathcal{H}_B$, where $\ket{\psi}$ is of finite Schmidt rank, and $\forall a,b,x,y$, $$ p(a,b|x,y) = \bra{\psi}A^a_x \otimes B^b_y \ket{\psi}
$$ The definition of $\mathcal{C}_{q}^{\infty,\infty,r,s}$ is analogous, except that $\mathcal{X}$, $\mathcal{Y}$ have countably infinite size, while $\mathcal{A}$, $\mathcal{B}$ have sizes $r$ and $s$.
\label{def: sets}
\end{definition}
 
Note that in the above definition $\mathcal{H}_A$ and $\mathcal{H}_B$ are allowed to be infinite-dimensional, but we require $\ket{\psi}$ to have finite Schmidt rank (i.e. finite entanglement). $\mathcal{C}_{qs}^{m,n,\infty,\infty}$ and $\mathcal{C}_{qs}^{\infty,\infty,r,s}$ are defined by simply dropping the requirement that $\ket{\psi}$ has finite Schmidt rank. We choose to work only with projective measurements for later convenience, but one could alternatively define $\mathcal{C}_{q}^{m,n,\infty,\infty}$ and $\mathcal{C}_{q}^{\infty,\infty,r,s}$ by restricting to finite-dimensional Hilbert spaces, and allowing the measurements to be infinite-outcome POVMs (Positive-Operator-Valued Measures). We show that these two definitions are equivalent. When it is clear from the context, we omit writing trivial identities on other subsystems: for example, we may write $A_x^a$ in place of $A_x^a \otimes \Id$.

\begin{lem}
The following are equivalent:
\begin{itemize}
\item[(i)] $p \in \mathcal{C}_{q}^{m,n,\infty,\infty}$ (according to Definition \ref{def: sets})
\item[(ii)] There exist finite-dimensional Hilbert spaces $\mathcal{H}_A$ and $\mathcal{H}_B$, a state $\ket{\psi} \in \mathcal{H}_A \otimes \mathcal{H}_B$ and POVMs $\{M^a_x\}_a$, $\{N^b_y\}_b$ on $\mathcal{H}_A$ and $\mathcal{H}_B$ respectively, such that $p(a,b|x,y) = \Tr{M^a_x N^b_y \ket{\psi}\bra{\psi}}$. 
\end{itemize}
\end{lem}

\begin{proof}
(ii) $\Rightarrow$ (i): Apply Naimark's dilation theorem \cite{Pau03} (note that it holds also for infinite-outcome POVMs).

(i) $\Rightarrow$ (ii): Let $\mathcal{H}_A$, $\mathcal{H}_B$ be (infinite-dimensional) Hilbert spaces, and $\ket{\psi} \in \mathcal{H}_A \otimes \mathcal{H}_B$ a bipartite state with finite Schmidt rank. Let $\{A^a_x\}_a$ and $\{B^b_y\}_b$ be infinite-outcome projective measurements such that $p(a,b|x,y) = \bra{\psi}A^a_x B^b_y\ket{\psi}$. Let $\mathcal{H}_{A'}$, $\mathcal{H}_{B'}$ be finite-dimensional Hilbert spaces with dimension the Schmidt rank of $\ket{\psi}$, and $\mathcal{H}_{\hat{A}}$, $\mathcal{H}_{\hat{B}}$ infinite-dimensional Hilbert spaces. Let $\ket{\psi'} \in \mathcal{H}_{A'} \otimes \mathcal{H}_{B'}$ be a state with the same Schmidt decomposition as $\ket{\psi}$, with respect to some basis of $\mathcal{H}_{A'} \otimes \mathcal{H}_{B'}$. Define isomorphisms $\Phi_D: \mathcal{H}_D \rightarrow \mathcal{H}_{D'} \otimes \mathcal{H}_{\hat{D}}$, for $D \in \{A,B\}$, such that $\Phi_A \otimes \Phi_B (\ket{\psi}) = \ket{\psi'}_{A'B'} \otimes \ket{00}_{\hat{A}\hat{B}}$, and the new projective measurements under the isomorphism $\{\tilde{A}^a_x\}_a$ and $\{\tilde{B}^b_y\}_b$. From these, we wish to obtain POVMs on just $\mathcal{H}_A'$ and $\mathcal{H}_B'$ such that, on $\ket{\psi'}$, they reproduce the correlation $p$. 

By hypothesis, $p(a,b|x,y) = \Tr{\ket{\psi'}\bra{\psi'} \otimes \ket{00}\bra{00} \tilde{A}^a_x \tilde{B}^b_y}$. Define $M^a_x := \pTr{\hat{A}}{I \otimes \ket{0}\bra{0}_{\hat{A}} \tilde{A}^a_x}$, and $N^b_y := \pTr{\hat{B}}{I \otimes \ket{0}\bra{0}_{\hat{B}} \tilde{B}^b_y}$. Then $p(a,b|x,y) = \Tr{\ket{\psi'}\bra{\psi'} M^a_x N^b_y }$. Moreover, one can check that $\{M^a_x\}_a$ and $\{N^b_y\}_b$ are POVMs on $\mathcal{H}_{A'}$ and $\mathcal{H}_{B'}$, as desired.
\end{proof}

\paragraph{Distance between correlations} We make precise the notion of distance between correlations.
\begin{definition}(Distance between correlations)
\label{def: distance}
Let $\{p(a,b|x,y): (a,b) \in \mathcal{A}\times \mathcal{B}\}_{(x,y) \in \mathcal{X}\times \mathcal{Y}}$ and $\{p'(a,b|x,y): (a,b) \in \mathcal{A}\times \mathcal{B}\}_{(x,y) \in \mathcal{X}\times \mathcal{Y}}$ be correlations on the same question and answer sets $\mathcal{X}, \mathcal{Y}, \mathcal{A}, \mathcal{B}$. Define their distance $|\cdot|_{corr}$ as 
\begin{equation}
|p-p'|_{corr} :=\sup_{x,y}   \sum_{a,b}|p(a,b|x,y)-p'(a,b|x,y)|
\end{equation}
\end{definition}

\paragraph{Self-testing} We define self-testing formally:
\begin{definition}[Self-testing]
We say that a correlation $\{p^*(a,b|x,y): a \in \mathcal{A}, b \in \mathcal{B}\}_{x \in \mathcal{X}, y \in \mathcal{Y}}$ self-tests a strategy $\left(\ket{\Psi}, \{\tilde{A}_{x}^{a}\}_{a}, \{\tilde{B}_{y}^{b}\}_{b} \right)$, with robustness $\delta(\epsilon)$, where $\delta(\epsilon) \rightarrow 0$, as $\epsilon \rightarrow 0$, if for any strategy $\left(\ket{\psi}, \{A^a_x\}_a, \{B^b_y\}_b \right)$ reproducing a correlation $p$ such that $|p-p^*|_{corr} \leq \epsilon$, there exists a local isometry $\Phi= \Phi_A\otimes\Phi_B$ such that 
\begin{align}
\| \Phi(\ket{\psi}) - \ket{\mathrm{extra}}\otimes \ket{\Psi} \| &\leq \delta(\epsilon) \label{eq: state}\\ \label{eq: measurements}
\| \Phi(A_{x}^{a}\otimes B_{y}^{b}\ket{\psi}) - \ket{\mathrm{extra}}\otimes (\tilde{A}_{x}^{a}\otimes \tilde{B}_{y}^{b} \ket{\Psi}) \| &\leq \delta(\epsilon), 
\end{align}
where $\ket{\mathrm{extra}}$ is some auxiliary state.
\end{definition}
Sometimes, we refer to \textit{self-testing of the state} when we are only concerned with the guarantee of equation \eqref{eq: state}, and not \eqref{eq: measurements}.

\paragraph{Tilted CHSH}

We briefly introduce the tilted CHSH inequality \cite{Acin12}, which is a building block for all of the correlations appearing in this work. Let $A_0, A_1, B_0, B_1$ be $\pm 1$-valued random variables. For a random variable $X$, let $\left<X\right>$ denote its expectation. The tilted CHSH inequality \cite{Acin12} is the following generalisation of the CHSH inequality:
\begin{equation}
    \left<\alpha A_0 + A_0B_0 + A_0B_1 +A_1B_0 - A_1B_1 \right > \leq 2+\alpha,
    \label{tiltedchsh}
\end{equation}
which holds when the random variables are local. The maximal quantum violation is $\sqrt{8+2\alpha^2}$ and is attained when the strategy of the two parties consists of sharing the joint state $\ket{\psi} = \cos \theta \ket{00} + \sin \theta \ket{11}$, and measuring observables $A_0, A_1$ and $B_0, B_1$ respectively, where $A_0 = \sigma_z$, $A_1 = \sigma_x$, $B_0 = \cos \mu \sigma_z + \sin \mu \sigma_x$ and $B_1 = \cos \mu \sigma_z + \sin \mu \sigma_x$, and $\sin 2\theta = \sqrt{\frac{4-\alpha^2}{4+\alpha^2}}$ and $\mu = \arctan \sin 2\theta$. The converse also holds, in the sense that maximal violation self-tests this strategy. This is made precise in the following lemma.

\begin{lem}[\cite{Bamps15}]
\label{Bamps lemma}
Let $\ket{\psi} \in \mathcal{H}_A \otimes \mathcal{H}_B$. Let $A_0, A_1$ and $B_0, B_1$ be binary observables, respectively on $\mathcal{H}_A$ and $\mathcal{H}_B$, with $\pm1$ eigenvalues. Suppose that
\begin{equation}
\bra{\psi}\alpha A_0 + A_0B_0 + A_0B_1 +A_1B_0 - A_1B_1 \ket{\psi} \geq \sqrt{8+\alpha^2} - \epsilon
\end{equation}
Let $\theta, \mu \in (0,\frac{\pi}{2})$ be such that $\sin 2\theta = \sqrt{\frac{4-\alpha^2}{4+\alpha^2}}$ and $\mu = \arctan \sin 2\theta$. Let $Z_A = A_0$, $X_A = A_1$. Let $Z^*_B$ and $X^*_B$ be respectively $\frac{B_0 + B_1}{2\cos \mu}$ and $\frac{B_0 - B_1}{2\sin \mu}$, but with all zero eigenvalues replaced by one. Define $Z_B = Z^*_B|Z^*_B|^{-1}$ and $X_B = X^*_B|X^*_B|^{-1}$. \\
Then, we have
\begin{align}
\|(Z_A - Z_B) \ps \| = O(\sqrt{\epsilon})& \\
\|\cos\theta X_A (\Id-Z_A) \ps - \sin\theta X_B (\Id+Z_B) \ps\| &= O(\sqrt{\epsilon})
\end{align}
\end{lem}

\section{Finite question sets and infinite answer sets}
\label{infinite answers}
We start by describing the bipartite quantum correlations that self-test any entangled pair of qudits, from \cite{CGS16}. We will then naturally extend these correlations to infinite answer sets (with the same question sets). The self-testing result for all finite-dimensional bipartite states from \cite{CGS16} will be a key ingredient in our proof that the new correlations can be achieved using an infinite-dimensional state, but not any finite-dimensional one. For the purpose of our proof, we will require a robust version of the result from \cite{CGS16}, of which we provide a proof in the Appendix.

\subsection{Correlations that self-test any entangled pair of qudits with questions sets of size $3$ and $4$, and answer sets of size $d$}
\label{2.1}

In this subsection, we present the correlations from \cite{CGS16} that self-test any entangled pair of qudits, for any finite $d$. The question sets are $\mathcal{X} = \{0,1,2\}$ and $\mathcal{Y} = \{0,1,2,3\}$, and the answer sets are $\mathcal{A}  = \mathcal{B} = \{0,1,\ldots ,d-1\}$ for Alice and Bob respectively. We start by describing ideal measurements that achieve the self-testing correlations, as we believe this aids understanding. Then, in Definition \ref{correlations many answers}, we describe properties of the self-testing correlations that are enough to characterize them, in the sense that any correlation satisfying these properties must be the self-testing correlation. Let $\sigma_Z$ and $\sigma_X$ be the usual Pauli matrices. For a single-qubit observable $A$, we denote by $[A]_m$ the observable defined with respect to the basis $\{\ket{2m},\ket{2m+1}\}$. For example, $[\sigma_Z]_m = \ket{2m}\bra{2m} - \ket{2m+1}\bra{2m+1}$. Similarly, we denote by $[A]'_m$ the observable defined with respect to the basis $\{\ket{2m+1}, \ket{2m+2}\}$. We use the notation $\bigoplus A_i$ to denote the direct sum of observables $A_i$.  We take $d$ to be odd, as this is the more relevant case to us. The case $d$ even is similar (and simpler).

\begin{definition}[Ideal measurements---many answers, finite case]
\label{ideal measurements - many answers}
For $d$ odd, let $\ket{\Psi} = \sum_{i=0}^{d-1} c_i \ket{ii}$ with $\sum_{i=0}^{d-1} c_i^2 = 1$. 

\begin{itemize}

\item For $x=0$: Alice measures in the computational basis (i.e. in the basis $\{\ket{0},\ket{1},\cdots,\ket{d-1}\}$). For $x=1$ and $x=2$, she measures in the eigenbases of observables 
\begin{align*}
&{\bigoplus_{m=0}^{\frac{d-1}{2}-1} [\sigma_X]_m\oplus \ket{d-1}\bra{d-1}} 
\text{ and } \\
&{\ket{0}\bra{0}\oplus\bigoplus_{m=0}^{\frac{d-1}{2}-1} [\sigma_X]'_m},
\end{align*}respectively, with the natural assignments of $d$ measurement outcomes.

\item In a similar way, for $y=0$ and $y=1$, Bob measures in the eigenbases of
\begin{align*}
&{\bigoplus_{m=0}^{\frac{d-1}{2}-1} [\cos{(\mu_m)}\sigma_Z+\sin{(\mu_m)}\sigma_X]_m\oplus\ket{d-1}\bra{d-1}} 
\text{ and }\\
&{\bigoplus_{m=0}^{\frac{d-1}{2}-1} [\cos{(\mu_m)}\sigma_Z-\sin{(\mu_m)}\sigma_X]_m\oplus\ket{d-1}\bra{d-1}},
\end{align*}
respectively, where $\mu_m = \arctan(\sin(2\theta_m))$ and $\theta_m = \arctan(\frac{c_{2m+1}}{c_{2m}})$. For $y=2$ and $y=3$, he measures in the eigenbases of
\begin{align*}
&{\ket{0}\bra{0}\oplus\bigoplus_{m=0}^{\frac{d-1}{2}-1} [\cos{(\mu'_m)}\sigma_Z+\sin{(\mu'_m)}\sigma_X]'_m}
\text{ and }\\
&{\ket{0}\bra{0}\oplus\bigoplus_{m=0}^{\frac{d-1}{2}-1} [\cos{(\mu'_m)}\sigma_Z-\sin{(\mu'_m)}\sigma_X]'_m},
\end{align*}
respectively, where $\mu'_m = \arctan(\sin(2\theta'_m))$ and $\theta'_m = \arctan(\frac{c_{2m+2}}{c_{2m+1}})$.

\end{itemize}
\end{definition}

The ideal measurements of Definition \ref{ideal measurements - many answers} define a correlation, which we refer to as the \textit{ideal correlation}. We now extract the essential properties of this correlation. The self-testing result from \cite{CGS16} then states that any correlation satisyfing these properties self-tests state $\ket{\Psi} = \sum_{i=0}^{d-1} c_i \ket{ii}$, as well as the ideal measurements, which implies that these properties are satisfied exclusively by the ideal correlation.

A convenient way to describe correlations is through correlation tables. A correlation can be specified by describing tables $T_{xy}$ for each possible question $(x,y) \in \mathcal{X} \times \mathcal{Y}$, with entries $T_{xy}(a,b) = p(a,b|x,y)$ for $(a,b)\in \mathcal{A} \times \mathcal{B}$. Let $\{T^{\text{tilted}}_{xy;\theta_m}\}_{x,y \in \{0,1\}}$ be the $2 \times 2$ correlation tables containing ideal tilted CHSH correlations self-testing the state $\cos{(\theta_m)}\ket{00}+\sin{(\theta_m)}\ket{11}$.

\begin{definition}[Self-testing properties of the ideal correlations---many answers, finite case]
\label{correlations many answers}

Let $\ket{\Psi} = \sum_{i=0}^{d-1} c_i \ket{ii}$, with $\sum_i c_i^2 =1 $. Take $d$ to be odd (this is the more relevant case for us, and the case $d$ even is similar). The self-testing properties of the ideal correlation for $\ket{\Psi}$ are: 
\begin{enumerate}[(i)]
\item For $x,y \in \{0,1\}$, $T_{xy}$ is block-diagonal with $2 \times 2$ blocks $C_{x,y,m}$ given by $(c_{2m}^2+c_{2m+1}^2)\cdot T^{\text{tilted}}_{xy;\theta_m}$, where $\theta_m :=\arctan \left(\frac{c_{2m+1}}{c_{2m}}\right) \in (0,\frac{\pi}{2})$. See table \ref{tab:txy1}.
\item For $x \in \{0,2\}$ and $y \in \{2,3\}$, $T_{xy}$ is also block-diagonal, but with blocks ``shifted down'' by one measurement outcome. Let the $2\times 2$ blocks be $D_{x,y,m}$ (corresponding to outcomes $2m+1$ and $2m+2$) for $x \in \{0,2\}$ and $y \in \{2,3\}$, defined as $D_{x,y,m}:= (c_{2m+1}^2+c_{2m+2}^2) \cdot T^{\text{tilted}}_{f(x),g(y); \theta_m'}$, where $\theta_m' :=\arctan\left(\frac{c_{2m+2}}{c_{2m+1}}\right) \in (0,\frac{\pi}{2})$, and $f(0) = 0, f(2) = 1, g(2) = 0, g(3) = 1$. See table \ref{tab:txy2}.
\end{enumerate}

\begin{table}[H]
\caption{$T_{xy}$ for $x,y\in \{0,1\}$, for $d$ odd}
\label{tab:txy1}
\begin{center}
\begin{tabular}{| c || c | c | c | c | c | c | c | c |}
	\hline
	$a \backslash b$ & 0 & 1 & 2 & 3 & $\cdots$ & $d-3$ & $d-2$ & d-1 \\ \hline \hline
	0 & \multicolumn{2}{| c |}{\multirow{2}{*}{$C_{x,y,m=0}$}} & 0 & 0 & $\cdots$ & 0 & 0 & 0\\ 
	\hhline{*{1}{|-}*{2}{|~}*{6}{|-}}
	1 & \multicolumn{2}{| c |}{} & 0 & 0 & $\cdots$ & 0 & 0 & 0\\ \hline
	2 & 0 & 0 & \multicolumn{2}{| c |}{\multirow{2}{*}{$C_{x,y,m=1}$}} & $\cdots$ & 0 & 0 & 0\\
	\hhline{*{3}{|-}*{2}{|~}*{4}{|-}}
	3 & 0 & 0 & \multicolumn{2}{| c |}{} & $\cdots$ & 0 & 0 & 0\\ \hline 	$\vdots$ & $\vdots$ & $\vdots$ & $\vdots$ & $\vdots$ & $\ddots$ & $\vdots$ & $\vdots$  & \vdots \\ \hline
	$d-3$ & 0 & 0 & 0 & 0 & $\cdots$ &  \multicolumn{2}{| c |}{\multirow{2}{*}{$C_{x,y,m=\frac{d-3}{2}}$}} & 0\\
	\hhline{*{6}{|-}*{2}{|~}*{1}{|-}}
	$d-2$ & 0 & 0 & 0 & 0 & $\cdots$ & \multicolumn{2}{| c |}{} & 0\\ \hline
	$d-1$ & 0 & 0 & 0 & 0 & $\cdots$ & 0 & 0 & $c_{d-1}^2$\\
	\hline
\end{tabular}
\end{center}
\end{table}
\begin{table}[H]
\caption{$T_{xy}$ for $x \in \{0,2\}, y \in \{2,3\}$}
\label{tab:txy2}
\begin{center}
\begin{tabular}{| c || c | c | c | c | c | c | c | c | c | c |}
	\hline
	$a \backslash b$ & 0 & 1 & 2 & 3 & 4 & $\cdots$ & $d-2$ & $d-1$\\ \hline \hline
    0 & $c_0^2$ & 0 & 0 & 0  & 0 & $\cdots$ & 0 & 0 \\ \hline
	1 & 0 & \multicolumn{2}{| c |}{\multirow{2}{*}{$D_{x,y,m=0}$}} & 0 & 0 & $\cdots$ & 0 & 0 \\
	\hhline{*{2}{|-}*{2}{|~}*{5}{|-}}
	2 & 0 & \multicolumn{2}{| c |}{} & 0 & 0 & $\cdots$ & 0 & 0\\ \hline
	3 & 0 & 0 & 0 & \multicolumn{2}{| c |}{\multirow{2}{*}{$D_{x,y,m=1}$}} & $\cdots$ & 0 & 0\\
	\hhline{*{4}{|-}*{2}{|~}*{3}{|-}}
	4 & 0 & 0 & 0 & \multicolumn{2}{| c |}{} & $\cdots$ & 0 & 0 \\ \hline
	$\vdots$ & $\vdots$ & $\vdots$ & $\vdots$ & $\vdots$ & $\vdots$ & $\ddots$  & $\vdots$ & $\vdots$ \\ \hline
    $d-2$ & 0 & 0 & 0 & 0 & 0 & $\cdots$ & \multicolumn{2}{| c |}{\multirow{2}{*}{$D_{x,y,m=0}$}} \\
	\hhline{*{7}{|-}*{2}{|~}}
    $d-1$ & 0 & 0 & 0 & 0 & 0 & $\cdots$ &\multicolumn{2}{| c |}{} \\ \hline
\end{tabular}
\end{center}
\end{table}

We refer the reader to \cite{CGS16} for an explicit presentation of the $2 \times 2$ blocks $C_{x,y,m}$.
\end{definition}

The following is a robust version of the self-testing result from \cite{CGS16}.

\begin{theorem} (\cite{CGS16})
\label{thm3}
For any bipartite entangled quantum state $\ket{\Psi} = \sum_{i=0}^{d-1} c_i \ket{ii}$, there exists a correlation $p^* \in \mathcal{C}_q^{3,4,d,d}$ (the one specified in Definition \ref{ideal measurements - many answers}) that self-tests $\ket{\Psi}$, with $O\left(d^3\epsilon^{\frac14}\right)$ robustness.
\end{theorem}
\noindent \textit{Proof.} Obtaining this (unoptimized) robustness bound is a straightforward adaption of the proof from \cite{CGS16}, and we include a proof in the Appendix for completeness. \\

\subsection{Correlations with finite question sets and infinite answer sets}
\label{2.2}
We are ready to present a correlation separating $\mathcal{C}_{q}^{3,4,\infty,\infty}$ and  $\mathcal{C}_{qs}^{3,4,\infty,\infty}$. Informally, this is defined as the limit of the correlations described in the previous subsection as the answer sets size tends to infinity, for some appropriate choice of $\ket{\Psi} = \sum_i c_i\ket{ii}$. We still have $\mathcal{X} = \{0,1,2\}$ and $\mathcal{Y} = \{0,1,2,3\}$, but now $\mathcal{A} = \mathcal{B} = \mathbb{N}$.

To make the definition rigorous, we introduce some notation. For any correlation $\{p(a,b|x,y)\}_{x,y}$ on finite question and answer sets $\mathcal{X}, \mathcal{Y}, \mathcal{A}, \mathcal{B}$, define its \textit{lift} to countably infinite answer sets to be the correlation $\{\hat{p}(a,b,x,y): (a,b) \in \mathbb{N}^2\}_{x,y}$, on the same question sets $\mathcal{X}, \mathcal{Y}$, such that, $\forall (x,y) \in \mathcal{X} \times \mathcal{Y}$, $\hat{p}(a,b,x,y) = p(a,b|x,y)$ for $(a,b) \in \mathcal{A} \times \mathcal{B}$, and $\hat{p}(a,b,x,y) = 0$ otherwise. 

From now on, we use $\{p^*_N(a,b|x,y)\}_{x,y}$ to refer to the ideal correlation from Definition \ref{ideal measurements - many answers}, specifically the one self-testing the state $\ket{\Psi_N}= C_N \cdot \sum_{i=0}^{N-1} \frac{1}{(i+1)^8} \ket{ii}$, for $N$ odd, where $C_N$ is a normalizing constant (and precisely $C_N = \sqrt{H_{N-1}^{(16)}}$, where $H_N^{(r)}:= \sum_{n=1}^N\frac1{n^r}$ are the generalized harmonic numbers.) The reader might wonder about the choice to have coefficients proportional to $\frac{1}{(i+1)^8}$ in the definition of $\ket{\Psi_N}$. The reason for this is that the choice of the coefficients in turn determines the rate at which the corresponding correlations $\hat{p}^*_N$ converge. In order for the proof of Theorem \ref{thm1} to work, we need the rate of convergence to be fast enough relative to the dependence on the local dimension in the robustness bound of our self-testing result from Theorem \ref{thm3} (more details on this in Section \ref{sec:proof_thm1}). 

\begin{definition}(Separating correlation, many answers)
\label{separating correlations defn}
We define the separating correlation in the \textit{many answers} case to be $p^*_{\infty} := \lim_{K \rightarrow \infty} \hat{p}^*_{2K+1}$, where the limit is defined pointwise.
\end{definition}
Notice that the limit is well-defined, since $\forall a,b,x,y$ the sequence $\big(\hat{p}_{2K+1}^*(a,b|x,y)\big)_K$ is easily seen to be convergent. 

For completeness we describe the ideal measurements achieving the separating correlations. We describe them in terms of generic coefficients $c_i$, although the particular choice made in Definition \ref{separating correlations defn} imposes $c_i =  C \cdot \frac{1}{(i+1)^8}$, where $C$ is a normalizing constant (precisely $C = \sqrt{\frac{3217 \pi^{16}}{325641566250}}$).
\begin{definition} (Ideal measurements for the separating correlation)
Let $\ket{\Psi_{\infty}} = \sum_{i=0}^{\infty} c_i \ket{ii}$, with $\sum_{i=0}^{\infty} c_i^2 = 1$ .

For $x=0$, Alice measures in the computational basis (i.e. in the basis $\{\ket{0},\ket{1}, \ldots\}$). For $x=1$ and $x=2$, she measures in the eigenbases of observables $\bigoplus_{m=0}^{\infty} [\sigma_x]_m$ and $\bigoplus_{m=0}^{\infty} [\sigma_x]'_m$ respectively, with the natural assignments of measurement outcomes.

In a similar way, for $y=0$ and $y=1$, Bob measures in the eigenbases of observables $\bigoplus_{m=0}^{\infty} [\cos{(\mu_m)}\sigma_z+\sin{(\mu_m)}\sigma_x]_m$ and $\bigoplus_{m=0}^{\infty} [\cos{(\mu_m)}\sigma_z-\sin{(\mu_m)}\sigma_x]_m$ respectively, with the natural assignments of measurement outcomes. Here $\mu_m = \arctan(\sin(2\theta_m))$, where $\theta_m  = \arctan(\frac{c_{2m+1}}{c_{2m}})$. For $y=2$ and $y=3$, he measures in the eigenbases of $\bigoplus_{m=0}^{\infty} [\cos{(\mu'_m)}\sigma_z+\sin{(\mu'_m)}\sigma_x]'_m$ and $\bigoplus_{m=0}^{\infty} [\cos{(\mu'_m)}\sigma_z-\sin{(\mu'_m)}\sigma_x]'_m$respectively, where $\mu'_m = \arctan(\sin(2\theta'_m))$ and $\theta'_m = \arctan(\frac{c_{2m+2}}{c_{2m+1}})$.

\end{definition}

It is straightforward to see that the ideal measurements above achieve $p^*_{\infty}$ (when $c_i = C \cdot \frac{1}{(i+1)^8}$).

\subsection{Proof of Theorem \ref{thm1}}
\label{sec:proof_thm1}
In this subsection, we prove Theorem \ref{thm1}.

\begin{claim}
\label{claim1}
There exists a function $\epsilon(N) = \alpha N^{-16} $, for some constant $\alpha$, such that 
\begin{equation}
|\{\hat{p}^*_N(a,b|x,y)\}_{x,y} - \{
p^*_{\infty}(a,b|x,y)\}_{x,y}|_{corr} \leq \epsilon(N) 
\end{equation}
\end{claim}
\begin{proof}
This is straightforward to see from the definitions of $\hat{p}^*_N$ and $p^*_{\infty}$ from subsection \ref{2.2}. In particular, the former is obtained by measuring $\ket{\Psi_N}= C_N \cdot \sum_{i=0}^{N-1} \frac{1}{(i+1)^8} \ket{ii}$, and the latter by measuring $\ket{\Psi_{\infty}}= C \cdot \sum_{i=0}^{\infty} \frac{1}{(i+1)^8} \ket{ii}$ with the same measurement settings, where $C_N$ and $C$ are the constants from subsection \ref{2.2}. It is clear that the trace distance $\|\ket{\Psi_N}\bra{\Psi_N}- \ket{\Psi_{\infty}}\bra{\Psi_{\infty}}\|_1$ shrinks as $O(N^{-16})$, and this implies the claim.
\end{proof}

\begin{proof}[Proof of Theorem \ref{thm1}]
We will show that, for any fixed dimension $N'$, $p^*_{\infty} \notin \mathcal{C}_{q \leq N'}^{3,4,\infty,\infty}$. This immediately implies that $p^*_{\infty} \notin \bigcup \limits_{N'} \mathcal{C}_{q \leq N'}^{3,4,\infty,\infty}= \mathcal{C}_{q}^{3,4,\infty,\infty}$, and, hence, $\mathcal{C}_{q}^{3,4,\infty,\infty} \neq \mathcal{C}_{qs}^{3,4,\infty,\infty}$.
The proof can be broken up into a few parts (described informally):
\begin{enumerate}
\item[(i)] The separating correlations $p^*_{\infty}$ are $O(N^{-16})$-close to $\hat{p}^*_N$. This is the content of Claim \ref{claim1}.
\item[(ii)] If $\ket{\psi}$ is a state achieving correlations $\delta$-close to $p^*_{\infty} $, robustness in the self-testing result of Theorem \ref{thm3} implies that $\|\ket{\psi} - \ket{\Psi_N}\| \leq O\big(N^3 (\alpha N^{-16} + \delta)^{\frac14}\big)$ up to a local isometry. 
\item[(iii)] For the latter to be true, $\ket{\psi}$ must have dimension $\Omega\big(\frac{1}{\delta^{32}}\big)$. 
\end{enumerate}
\vspace{0.1cm}
We describe the above steps formally.

Suppose there exists $p \in \mathcal{C}_{q \leq N'}^{3,4,\infty,\infty}$ such that $|p - p^*_{\infty}|_{corr} = \delta$, for some $\delta > 0$ and some dimension $N'$. Then, by a triangle inequality using Claim \ref{claim1}, we have that for all $N$
\begin{equation}
\label{eq7}
|p  - \hat{p}^*_N|_{corr} \leq \epsilon(N) + \delta
\end{equation}

Now, define a new correlation $p_N$ obtained from $p$ by simply classically post-processing the outcomes for the $N'$-dimensional quantum strategy achieving $p$ so that each of Alice and Bob maps outcomes in $\mathbb{N}\setminus\{0,\ldots ,N-1\}$ to outcome $0$. Clearly, then, $p_N$ is still an $N'$-dimensional quantum correlation, and we can view it either as a strategy in $ \mathcal{C}_{q \leq N'}^{3,4,N,N}$ or in $\mathcal{C}_{q \leq N'}^{3,4,\infty,\infty}$ with zero probability mass on the outcomes outside of $\{0,\ldots ,N-1\}^2$. To be precise, we denote the former by $p_N$ and the latter by $\hat{p}_N$. Moreover, notice that $ \forall (x,y)$, $\sum_{a,b \in \mathbb{N}^2\setminus\{0,\ldots ,N-1\}^2}p(a,b|x,y) \leq \epsilon(N) + \delta$, by $\eqref{eq7}$. Hence, it's easy to see that $|\hat{p}_N  - p|_{corr} \leq 2(\epsilon(N) + \delta)$. Then, by a triangle inequality, 
\begin{align}
&|\hat{p}_N  - \hat{p}^*_N|_{corr} \leq 3\epsilon(N) + 3\delta \\
\Rightarrow \,\,\,& |p_N - p^*_N|_{corr} \leq 3\epsilon(N) + 3\delta \label{eq15}
\end{align}

And this holds for all $N$. Now, denote by $\ket{\psi}$ the state, of Schmidt rank at most $N'$, that achieves the correlation $p_N$. Then, by Theorem \ref{thm3}, there exists a family of local isometries $\{\Phi_N\}$ and auxiliary states $\{\ket{extra_N}\}$ such that
\begin{equation}
\|\Phi_N(\ket{\psi}) - \ket{\Psi_N} \otimes \ket{extra_N}\| = O\Big(N^3 \big(\epsilon(N)+ \delta\big)^{\frac14}\Big) = O\left(N^3\Big(\alpha N^{-16} + \delta\Big)^{\frac14}\right) \label{eq9}
\end{equation}
Notice that since $\Phi_N$ is a local isometry, then $\Phi_N(\ket{\psi})$ has Schmidt rank at most $N'$, while $\ket{\Psi_N} \otimes \ket{extra}$ has Schmidt rank at least $N$, with Schmidt coefficients $C_N \cdot \frac{1}{(i+1)^8} \cdot \alpha_j$, where the $\alpha_j$ are the Schmidt coefficients of $\ket{extra}$.
This implies that
\begin{equation}
\label{eq:13}
\| \Phi_N(\ket{\psi}) - \ket{\Psi_N} \otimes \ket{extra_N} \| =\Omega(N'^{-8})
\end{equation}

Now, choose $N \approx \delta^{-\frac{1}{16}}$ in \eqref{eq9}. This gives
\begin{equation}
\label{eq:14}
\|\Phi_N(\ket{\psi}) - \ket{\Psi_N} \otimes \ket{extra_N}\| = O( \delta^{-\frac{1}{48}}\cdot \delta^{\frac14})
\end{equation}

The only way for equations \eqref{eq:13} and \eqref{eq:14} to be compatible is that $N' = \Omega\big(\delta^{-\frac{1}{32}}\big)$, which completes the proof of Theorem \ref{thm1}.
\end{proof}

\section{Infinite question sets and finite answer sets}
\label{infinite questions}
We turn to the case of infinite question sets and finite answer sets. We start by presenting novel correlations that self-test any entangled pair of qudits for any finite dimension $d$. 
Then we extend these to correlations on infinite question sets that give us the desired separation between finite and infinite-dimensional quantum correlations.

\subsection{Self-testing all pure bipartite entangled states with $O(d)$ measurements and $\leq 4$ outcomes per party}
\label{new self-test}

Following the structure of section \ref{infinite answers}, we describe the self-testing correlation by first presenting the ideal state and ideal measurements that achieve it. These consist of $\frac32 d$ and $\frac52 d$ measurements with $4$ and $3$ outcomes for Alice and Bob respectively. Afterward, we define the essential properties of the ideal correlation arising from these measurements, which are enough to characterize it. Our main result will be that the ideal correlation self-tests the ideal state. The question sets are $\mathcal{X} = \{0,1,...\frac{d}{2}-1\} \times\{Z,X,X'\}$, $\mathcal{Y} = \{0,\ldots ,\frac{d}{2}-1\} \times\{Z,Z',X,X', \mbox{Aux}\}$, and the answer sets are $\mathcal{A} = \{0,1,2, \perp\}$, $\mathcal{B} = \{0,1, \perp\}$.

\begin{definition}[Ideal measurements---many questions, finite case]
\label{ideal measurements many questions finite case}
Assume $d$ is even, the case $d$ odd being similar. Let $\ket{\Psi} = \sum_{i=0}^{d-1} c_i \ket{ii}$ with $\sum_{i=0}^{d-1} c_i^2 = 1$.

\begin{itemize}

\item For $m=0,\ldots ,\frac{d}{2}-1$ and $x = (m,Z)$, Alice performs the projective measurement 
\begin{equation}
\left\{\ket{2m}\bra{2m}, \ket{2m+1}\bra{2m+1}, \ket{2m+2}\bra{2m+2}, \Id - P_1\right\},
\end{equation}
where $P_1$ is the sum of the first three projections. 
She assigns measurement outcomes $0,1,2,\perp$, respectively to the four projectors above in the order they are listed.

For $m=0,\ldots ,\frac{d}{2}-1$ and $x = (m,X)$, Alice performs the projective measurement 
\begin{equation}
\left\{\frac{\ket{2m}+\ket{2m+1}}{\sqrt{2}}\frac{\bra{2m}+\bra{2m+1}}{\sqrt{2}}, \frac{\ket{2m}-\ket{2m+1}}{\sqrt{2}}\frac{\bra{2m}-\bra{2m+1}}{\sqrt{2}}, \, \ket{2m+2}\bra{2m+2},\, \Id - P_2\right\},
\end{equation}
where $P_2$ is the sum of the first three projections. 
She assigns measurement outcomes $0,1,2,\perp$, respectively.

For $m=0,\ldots ,\frac{d}{2}-1$ and $x = (m,X')$, Alice performs the projective measurement 
\begin{equation}
\left\{\ket{2m}\bra{2m}, \frac{\ket{2m+1}+\ket{2m+2}}{\sqrt{2}}\frac{\bra{2m+1}+\bra{2m+2}}{\sqrt{2}}, \frac{\ket{2m+1}-\ket{2m+2}}{\sqrt{2}}\frac{\bra{2m+1}-\bra{2m+2}}{\sqrt{2}},\, \Id - P_3\right\}, 
\end{equation}
where $P_3$ is the sum of the first three projections.
She assigns measurement outcomes $0,1,2,\perp$, respectively.

\item Bob, instead, for $m=0,\ldots ,\frac{d}{2}-1$ and $y = (m,Z)$ performs the following measurement:
Two projectors are onto the $\pm 1$ eigenvectors of $[\cos{(\mu_m)}\sigma_z+\sin{(\mu_m)}\sigma_x]_m$, where $\mu_m = \arctan (\sin 2\theta_m)$ and $\theta_m  = \arctan(\frac{c_{2m+1}}{c_{2m}})$. To these, Bob assigns respectively outcomes $0$ and $1$. The third projector is on everything else, and corresponds to outcome $\perp$. 

For $m=0,\ldots ,\frac{d}{2}-1$ and $y = (m,X)$, the same but with $[\cos{(\mu_m)}\sigma_z-\sin{(\mu_m)}\sigma_x]_m$.

For $m=0,\ldots ,\frac{d}{2}-1$ and $y = (m,Z'), (m,X')$, Bob's measurements are the same as above, except with $[\cos{(\mu'_m)}\sigma_z+\sin{(\mu'_m)}\sigma_x]_m$ and $[\cos{(\mu'_m)}\sigma_z-\sin{(\mu'_m)}\sigma_x]_m$ respectively, where $\mu_m' = \arctan (\sin 2\theta_m')$, with $\theta'_m  = \arctan(\frac{c_{2m+2}}{c_{2m+1}})$.

For $m=0,\ldots ,\frac{d}{2}-1$ and $y = (m, \mbox{Aux})$, Bob performs the projective measurement
\begin{equation}
\left\{\ket{2m}\bra{2m},\ket{2m+1}\bra{2m+1}, \, \Id - P'\right\},
\end{equation}
where $P'$ is the sum of the first two projections. He assigns measurement outcomes $0,1 \perp$ respectively.

\end{itemize}
\end{definition}


We refer to the correlation arising from the ideal measurements of Definition \ref{ideal measurements many questions finite case} as the ideal correlation. Now, we extract the essential self-testing properties of the correlation resulting from the ideal measurements. Stating them concisely will aid the proof of Theorem \ref{thm4}. Recall that $T^{\text{tilted}}_{ij;\theta}$, for $i,j \in \{0,1\}$, are the $2 \times 2$ correlation tables which correspond to the maximal violation of the tilted-CHSH inequality which self-tests the state $\cos{\theta}\ket{00}+\sin{\theta}\ket{11}$. 

\begin{definition}[Self-testing properties---many questions, finite case]
\label{correlations many questions finite case}
Let $\ket{\Psi} = \sum_{i=0}^{d-1} c_i \ket{ii}$, with $\sum_i c_i^2 = 1$. The self-testing properties of the ideal correlation for $\ket{\Psi}$ are:
\begin{enumerate}[(i)]
\item 
\label{many-question-self-test-criterion-1}
For $m=0,\ldots ,\frac{d}{2}-1$ and $x,y \in \{m\} \times \{Z,X\}$ the $4\times 4$ table $T_{xy}$ has the following form:
\begin{table}[H]
\caption{$T_{xy}$ for $x,y \in \{m\} \times \{Z,X\}$}
\label{tab:txy7}
\begin{center}
\begin{tabular}{| c || c | c | c | c | c |}
	\hline
	$a \backslash b$ & 0 & 1 & $\perp$ \\ \hline \hline
	0 & \multicolumn{2}{| c |}{\multirow{2}{*}{$C_{x,y,m}$}} & 0 \\ 
	\hhline{*{1}{|-}*{2}{|~}*{3}{|-}}
	1 & \multicolumn{2}{| c |}{} & 0  \\ \hline
	2 & 0 & 0 & *\\
	\hline
	$\perp$ & 0 & 0 & * \\ \hline
\end{tabular}
\end{center}
\end{table}
Here, define $f:\mathcal{X}\rightarrow \{0,1\}$, $g:\mathcal{Y} \rightarrow \{0,1\}$ so that $f\left((\cdot,Z)\right) = g\left((\cdot,Z)\right) = g\left((\cdot,Z')\right) = 0$ and $f\left((\cdot,X)\right) = f\left((\cdot,X')\right)= g\left((\cdot,X)\right) = g\left((\cdot,X')\right) = 1$. Then, $C_{x,y,m}$ is given by $(c_{2m}^2+c_{2m+1}^2)\cdot T^{\text{tilted}}_{f(x)f(y);\theta_m}$, where $\theta_m :=\arctan \left(\frac{c_{2m+1}}{c_{2m}}\right) \in (0,\frac{\pi}{2})$.

\item 
\label{many-question-self-test-criterion-2}
For $m=0,\ldots ,\frac{d}{2}-1$ and $x \in \{m\} \times \{Z,X'\},  y \in \{m\} \times \{Z',X'\}$ the $4\times 4$ table $T_{xy}$ has the form: 
\begin{table}[H]
\caption{$T_{xy}$ for $x \in \{m\} \times \{Z,X'\},  y \in \{m\} \times \{Z',X'\}$}
\label{tab:txy8}
\begin{center}
\begin{tabular}{| c || c | c | c | c | c |}
	\hline
	$a \backslash b$ & 0 & 1 & $\perp$ \\ 
    \hline \hline
	0 & 0 & 0 & *\\ \hline
	1 & \multicolumn{2}{| c |}{\multirow{2}{*}{$D_{x,y,m}$}} & 0 \\ 
	\hhline{*{1}{|-}*{2}{|~}*{3}{|-}}
	2 & \multicolumn{2}{| c |}{} & 0  \\
	\hline
	$\perp$ & 0 & 0 & * \\ \hline
\end{tabular}
\end{center}
\end{table}
where $D_{x,y,m}:= (c_{2m+1}^2+c_{2m+2}^2) \cdot C^{ideal}_{f(x),g(y); \theta_m'}$, where $\theta_m' :=\arctan\left(\frac{c_{2m+2}}{c_{2m+1}}\right) \in (0,\frac{\pi}{2})$.

\item 
\label{many-question-self-test-criterion-3}
For $m=0,\ldots,\frac{d}{2}-2$, we have
\begin{align}
p\left(a=2|x=(m,Z)\right) = p\left(a=0|x=(m+1,Z)\right) = p\left(b=0|y=(m+1,\mbox{Aux})\right) &= c_{2m+2}^2,
\\
p\left(2,0|(m,Z), (m+1,\mbox{Aux})\right) = p\left(0,0|(m+1,Z), (m+1,\mbox{Aux})\right) &= c_{2m+2}^2.
\end{align}
\item 
\label{many-question-self-test-criterion-4}
$\forall a,b \in \{0,1\}$ and $m\neq m'$, $p\left(a,b|(m,Z),(m',\mbox{Aux})\right) = 0$.

\end{enumerate}
\end{definition}


\vspace{5mm}

The following is our self-testing result.
\begin{theorem}
\label{thm4}
The ideal correlation from Definition \ref{ideal measurements many questions finite case} self-tests the state $\ket{\Psi} = \sum_{i=0}^{d-1}c_i \ket{ii}$, with $O(d^3\epsilon^{\frac14})$ robustness.
\end{theorem}
The proof is an adaptation of the proof, from \cite{CGS16}, that all pure bipartite entangled states can be self-tested using $3$ and $4$ measurement settings respectively for Alice and Bob and $d$-outcome measurements. The proof will occupy the rest of this subsection, and will proceed by showing the existence of unitary operators satisfying a robust (and slightly more general) version of the Yang--Navascu\'es self-testing criterion from \cite{Yang13}, which we state as the following lemma.

\begin{lem}
\label{YNcriterion}
Let $\ket{\Psi} = \sum_{i=0}^{d-1} c_i \ket{ii}$, where $0< c_i < 1$ for all $i$ and $\sum_{i=0}^{d-1}c_i^2 = 1$. Suppose there exist unitary operators $X^{(k)}_{A}, X^{(k)}_{B}$ and projections (not necessarily orthogonal) $\{P^{(k)}_{A}\}_{k=0,\ldots ,d-1}$ and $\{P^{(k)}_{B}\}_{k=0,\ldots ,d-1}$ satisfying the following conditions:
\begin{align}
\label{condition:YN-criterion-1}
\| P_A^{(i)}P_A^{(j)}\ket{\psi} \| &\leq \epsilon  \,\,\, \forall i\neq j, \\
\label{condition:YN-criterion-2}
\norm{\left(\sum_k P_A^{(k)} - \Id\right) \ps} &\leq \epsilon,  \\ 
\label{condition:YN-criterion-3}
\norm{ (P^{(k)}_{A} - P^{(k)}_{B} )\ket{\psi} } &\leq \epsilon \;\; \forall k,  \\
\label{condition:YN-criterion-4}
\norm{ (X^{(k)}_{A}X^{(k)}_{B}P^{(k)}_{A} - \frac{c_k}{c_0} P^{(0)}_{A}) \ket{\psi} } &\leq \epsilon 
\end{align}
Then there exists a local isometry $\Phi$ such that 
\begin{equation}\| \Phi(\ket{\psi}) - \ket{\text{extra}}\otimes\ket{\Psi} \| = O(d^{\frac52}\epsilon^{\frac12}).
\end{equation}
\end{lem}
\vspace{5mm}

\noindent \textit{Proof of Lemma \ref{YNcriterion}.}
The proof is included in the Appendix.

\begin{proof}[Proof of Theorem \ref{thm4}]
We present the proof for the exact case, and then robustness is argued analogously to Theorem \ref{thm3}. Suppose a strategy of Alice and Bob achieves the ideal correlation of definition \ref{ideal measurements many questions finite case}. Let this be described by a joint state $\ps$ and projectors $\Pi_{A_{x}}^a$ ($\Pi_{B_{y}}^b$) corresponding to Alice (Bob) obtaining outcome $a$ $(b)$ on question $x$ $(y)$.
\paragraph{Construction of the projections and ``flip'' operators}
In this paragraph, we construct the projections of Lemma \ref{YNcriterion}, as well as ``flip'' operators $X^{\mathscr{u}}_{A,m}, X^{\mathscr{u}}_{B,m}, Y^{\mathscr{u}}_{A,m}$ and $Y^{\mathscr{u}}_{B,m}$, satisfying equations \eqref{eq21} and \eqref{flip2}. The name refers to the fact that, informally, they ``flip'' a projection $P^{(2m+1)}_A$ to $P^{(2m)}_A$ and $P^{(2m+2)}_A$ to $P^{(2m+1)}_A$. We will then construct the unitaries from the conditions of Lemma \ref{YNcriterion} as appropriate alternating products of $X$'s and $Y$'s.
Define 
\begin{align}
\hat{A}_{m,Z} &=  \Pi_{A_{(m,Z)}}^0 - \Pi_{A_{(m,Z)}}^1, 
&\hat{B}_{m,Z} &=  \Pi_{B_{(m,Z)}}^0 - \Pi_{B_{(m,Z)}}^1,
\\\hat{A}_{(m,X)} &=  \Pi_{A_{(m,X)}}^0 - \Pi_{A_{(m,X)}}^1\text{, and} 
&\hat{B}_{m,X} &=  \Pi_{B_{(m,X)}}^0 - \Pi_{B_{(m,X)}}^1.
\end{align}
Then, let
 \begin{align}
 \Id_{A_{m,W}}:=\Pi_{A_{(m,W)}}^0 + \Pi_{A_{(m,W)}}^1= (\Ah_{m,W})^2 , 
 &&
 \Id_{B_{m,W}}:=\Pi_{B_{(m,W)}}^0 + \Pi_{B_{(m,W)}}^1 =(\Bh_{m,W})^2 
 \end{align}
  for $W \in \{Z,X\}$.

From property (i) of the ideal correlation from Definition \ref{correlations many questions finite case}, we obtain
\begin{align}
\|\Pi_{A_{(m,Z)}}^0 \ps\| &= \sqrt{\average{\psi| \Pi_{A_{(m,Z)}}^0| \psi}} = \sqrt{\average{\psi| \Pi_{A_{(m,Z)}}^0 \cdot \sum_{b=0}^{1}\Pi_{B_{(m,Z)}}^b |\psi}} 
\\&= \sqrt{c_{2m}^2\cos^2{(\frac{\mu_m}{2})} + c_{2m}^2\sin^2{(\frac{\mu_m}{2})}} = c_{2m}, 
\end{align}
and similarly $\|\Pi_{A_{(m,Z)}}^1 \ps\| =  c_{2m+1}$. With similar other calculations we deduce that
\begin{equation}
\label{id_norms}
\|\Id_{A_{m,W}} \ps \| = \|\Id_{B_{m,\tilde{W}}} \ps\| = \sqrt{c_{2m}^2 + c_{2m+1}^2}\text{ for $W,\tilde{W} \in \{Z,X\}\,$.}
\end{equation}
Moreover, notice that $\average{\psi| \Id_{A_{m,W}}\Id_{B_{m,\tilde{W}}}| \psi} = c_{2m}^2 + c_{2m+1}^2 =  \|\Id_{A_{m,W}}\ps \| \cdot \|\Id_{B_{m,\tilde{W}}}\ps \|$. Hence, by Cauchy-Schwarz, it must be the case that \begin{equation}
\label{id_equalities}
\Id_{A_{m,W}} \ps = \Id_{B_{m,\tilde{W}}}  \ps \,\,\,\, \mbox{for $W,\tilde{W} \in \{Z,X\}\,$}.
\end{equation}

By design, property (i) of Definition \ref{correlations many questions finite case} implies that 
\begin{equation}
\bra{\psi}\alpha_m\hat{A}_{m,Z}+\hat{A}_{m,Z}\hat{B}_{m,Z}+\hat{A}_{m,Z}\hat{B}_{m,X}+\hat{A}_{m,X}\hat{B}_{m,Z}-\hat{A}_{m,X}\hat{B}_{m,X}\ket{\psi}=\sqrt{8+2\alpha_m^2} \cdot (c_{2m}^2 + c_{2m+1}^2) \label{almost_tilted_chsh}
\end{equation}
where $\alpha_m=\frac{2}{\sqrt{1+2\tan^2{(2\theta_m)}}}$. As such, this is not a maximal violation of the tilted CHSH inequality (since $\ps$ has unit norm). We get around this by defining the normalized state $\ket{\psi_m}=\frac{\Id_{A_{m,Z}}\ket{\psi}}{\sqrt{c_{2m}^2+c_{2m+1}^2}}$. By \eqref{id_equalities},
\begin{align}
\hat{A}_{m,W} \ps &= \hat{A}_{m,W} \Id_{A_{m,W}} \ps = \hat{A}_{m,W}\Id_{A_{m,Z}}  \ps\text{, and}
\\
\hat{B}_{m,W} \ps &= \hat{B}_{m,W} \Id_{B_{m,W}} \ps = \hat{B}_{m,W}\Id_{A_{m,Z}} \ps.
\end{align}
Then \eqref{almost_tilted_chsh} implies
\begin{equation}
\label{eq16s}
\bra{\psi_m}\alpha_m\hat{A}_{m,Z}+\hat{A}_{m,Z}\hat{B}_{m,Z}+\hat{A}_{m,Z}\hat{B}_{m,X}+\hat{A}_{m,X}\hat{B}_{m,Z}-\hat{A}_{m,X}\hat{B}_{m,X}\ket{\psi_m}=\sqrt{8+2\alpha_m^2}.
\end{equation}
Define unitaries $\hat{A}^{\mathscr{u}}_{m,W} := \Id-\Id_{A_{m,W}} + \hat{A}_{m,W}$ and $\hat{B}^{\mathscr{u}}_{m,W} := \Id-\Id_{B_{m,W}} + \hat{B}_{m,W}$ for $W \in \{Z,X\}$. We think of these as the ``unitarized'' versions of the operators in \eqref{eq16s}. It is clear that equation \eqref{eq16s} holds also with the unitarized operators. Now let $Z^{\mathscr{u}}_{A,m} := \hat{A}^{\mathscr{u}}_{m,Z}$, $X^{\mathscr{u}}_{A,m} := \hat{A}^{\mathscr{u}}_{m,X}$. Then, let $Z^*_{B,m}$ and $X^*_{B,m}$ be $\frac{\Bh^{\mathscr{u}}_{m,Z}+ \Bh^{\mathscr{u}}_{m,X}}{2\cos(\mu_m)}$ and $\frac{\Bh^{\mathscr{u}}_{m,Z}- \Bh^{\mathscr{u}}_{m,X}}{2\sin(\mu_m)}$ respectively, but with all $0$ eigenvalues replaced by $1$. Define $Z^{\mathscr{u}}_{B,m} = Z^*_{B,m}|Z^*_{B,m}|^{-1}$ and $X^{\mathscr{u}}_{B,m} = X^*_{B,m}|X^*_{B,m}|^{-1}$ (this is again a required unitarization step). Then, by Lemma \ref{Bamps lemma}, the above maximal violation of the tilted-CHSH inequality implies that
\begin{align}
Z^{\mathscr{u}}_{A,m} \ket{\psi_m} &= Z^{\mathscr{u}}_{B,m} \ket{\psi_m}\text{, and } \label{44}\\
X^{\mathscr{u}}_{A,m} (\Id - Z^{\mathscr{u}}_{A,m}) \ket{\psi_m} &= \tan(\theta_m) X^{\mathscr{u}}_{B,m} (\Id + Z^{\mathscr{u}}_{A,m}) \ket{\psi_m}. \label{52}
\end{align}


Define the subspace $\mathcal{B}_m=\mbox{range}(\Id_{B_{m,Z}}) + \mbox{range}(\Id_{B_{m,X}})$, and the projection $\Id_{\mathcal{B}_m}$ onto subspace $\mathcal{B}_m$. Notice that $Z^{\mathscr{u}}_{B,m} = \Id-\Id_{\mathcal{B}_m} + \tilde{Z}_{B,m}$, where $\tilde{Z}_{B,m}$ is some operator supported only on subspace $\mathcal{B}_m$. This implies that $Z^{\mathscr{u}}_{B,m} \ket{\psi_m} =  \tilde{Z}_{B,m} \ket{\psi_m} = \tilde{Z}_{B,m} \ps$, where we have used \eqref{id_equalities} and the fact that 
\begin{align}\label{onesb}
    \Id_{B_{m,Z}}  \ps = \Id_{B_{m,X}}  \ps
    &&\text{implies}
    && \Id_{\mathcal{B}_m} \ket{\psi} = \Id_{B_{m,W}}\ket{\psi},
    && W \in \{Z,X\}.
\end{align}

Hence, from \eqref{44} we deduce that $\hat{A}_{m,Z} \ps= \tilde{Z}_{B,m} \ps$. For $m \in \{0,1,...\frac{d}{2}-1\}$, define projections 
\begin{align}
P_A^{(2m)} &:= (\Id_{A_{m,Z}} + \hat{A}_{m,Z})/2 = \Pi_{A_{(m,Z)}}^0,
&P_B^{(2m)} &:= (\Id_{\mathcal{B}_m} + \tilde{Z}_{B,m})/2,
\\P_A^{(2m+1)} &:= (\Id_{A_{m,Z}}  - \hat{A}_{m,Z})/2 = \Pi_{A_{(m,Z)}}^1,
&P_B^{(2m+1)} &:= (\Id_{\mathcal{B}_m} - \tilde{Z}_{B,m})/2. 
\end{align}

Note that $P_B^{(2m)},P_B^{(2m+1)}$ are indeed projections, since $\tilde{Z}_{B,m}$ has all $\pm 1$ eigenvalues corresponding to subspace $\mathcal{B}_m$, and is zero outside. We also have, for all $m$ and $k=2m,2m+1$,
\begin{align}
\label{eq16}
P_A^{(k)} \ps = (\Id_m^{A_0} +(-1)^k \hat{A}_{m,Z})/2 \ps &=  (\Id_m^{B_0} +(-1)^k \hat{A}_{m,Z})/2 \ps \nonumber\\
&= (\Id_{\mathcal{B}_m} +(-1)^k \tilde{Z}_{B,m})/2 \ps = P_B^{(k)} \ps
\end{align}
Further, notice that 
\begin{equation}
(\Id +(-1)^k Z^{\mathscr{u}}_{A,m}) \ket{\psi_m} =  (\Id_m^{A_0} +(-1)^k\hat{A}_{0,m}) \ket{\psi_m} = (\Id_m^{A_0} +(-1)^k \hat{A}_{0,m}) \ps = P_A^{(k)}\ps. 
\end{equation}

Combining with Equation \eqref{52} gives
\begin{equation}
\label{eq21}
X^{\mathscr{u}}_{A,m} P_A^{(2m+1)} \ps = \tan(\theta_m) X^{\mathscr{u}}_{B,m} P_A^{(2m)} \ps  = \frac{c_{2m+1}}{c_{2m}} X^{\mathscr{u}}_{B,m} P_A^{(2m)} \ps.\\
\end{equation}
\vspace{5mm}

Now, we can repeat an analogous procedure but starting from property (ii) of the ideal correlation from Definition \ref{correlations many questions finite case}, to deduce the existence of unitary operators $Y^{\mathscr{u}}_{A,m}, Y^{\mathscr{u}}_{B,m},$ satisfying 
\begin{equation}
\label{eq22}
Y^{\mathscr{u}}_{A,m} \Pi_{A_{(m,Z)}}^2 \ps = \tan(\theta'_m) Y^{\mathscr{u}}_{B,m} P_A^{(2m+1)} \ps  = \frac{c_{2m+2}}{c_{2m+1}} Y^{\mathscr{u}}_{B,m} P_A^{(2m+1)} \ps\\
\end{equation}
Notice, importantly, that the LHS involves $\Pi_{A_{(m,Z)}}^2$, and not $P_A^{(2m+2)}= \Pi_{A_{(m+1,Z)}}^0$. We would like to replace $\Pi_{A_{(m,Z)}}^2$ with $P_A^{(2m+2)}$, and for this we need property (iii) of Definition \ref{correlations many questions finite case}. This tells us that, for $m=0,..,\frac{d}{2}-2$,
\begin{align}
\label{eq23}
\sandwich{\psi}{\Pi_{A_{(m,Z)}}^2}{\psi}
= \sandwich{\psi}{\Pi_{A_{(m+1,Z)}}^0}{\psi} 
= \sandwich{\psi}{\Pi_{B_{(m+1,\text{Aux})}}^0}{\psi} 
&= c_{2m+2}^2\text{, and}
\\
\label{eq24}
\sandwich{\psi}{\Pi_{A_{(m,Z)}}^2\Pi_{B_{(m+1,\text{Aux})}}^0}{\psi} 
= \sandwich{\psi}{\Pi_{A_{(m+1,Z)}}^0\Pi_{B_{(m+1,\text{Aux})}}^0}{\psi} 
&=  c_{2m+2}^2.
\end{align}
\eqref{eq23} and \eqref{eq24} imply, with an application of Cauchy-Schwarz, that 
\begin{equation}
\label{eq25}
\Pi_{A_{(m,Z)}}^2 \ps = \Pi_{B_{(m+1,\text{Aux})}}^0 \ps= \Pi_{A_{(m+1,Z)}}^0 \ps 
\end{equation}
i.e. $\Pi_{A_{(m,Z)}}^2 \ps = P^{(2m+2)}_A \ps$. Plugging this into \eqref{eq22} gives 
\begin{equation}
\label{flip2}
Y^{\mathscr{u}}_{A,m} P^{(2m+2)}_A \ps = \frac{c_{2m+2}}{c_{2m+1}} Y^{\mathscr{u}}_{B,m} P_A^{(2m+1)} \ps\\
\end{equation}
\paragraph{The projections satisfy the Yang--Navascu\'es criterion}
So far, we have constructed sets of projections $P_{A}^{(k)},P_{B}^{(k)}$ for which $P_A^{(k)} \ps = P_B^{(k)}\ps$.
We also need them to satisfy conditions \eqref{condition:YN-criterion-1} and \eqref{condition:YN-criterion-2}. For this, we use property (iv) from Definition \ref{correlations many questions finite case}, which reads
\begin{equation}
\label{eq51}
\forall a,b \in \{0,1\}, \forall m \neq m', \,\,\, \,\bra{\psi}\Pi_{A_{(m,Z)}}^a\Pi_{B_{(m',\text{Aux})}}^b\ps = 0
\end{equation}
First, note that $P_A^{(2m)}P_A^{(2m+1)}\ps = 0 \,\, \forall m$. Moreover, \eqref{eq51} implies that for $k = 2m, k'= 2m'$, with $m\neq m'$, we have
\begin{align}
\|P_A^{(k)}P_A^{(k')}\ps \|^2 &= \bra{\psi}\Pi_{A_{(m,Z)}}^0\Pi_{A_{(m',Z)}}^0\Pi_{A_{(m,Z)}}^0 \ps \\
&= \bra{\psi}\Pi_{B_{(m,\text{Aux})}}^0\Pi_{A_{(m',Z)}}^0\Pi_{A_{(m,Z)}}^0 \ps \\
&= \bra{\psi}\Pi_{A_{(m',Z)}}^0\Pi_{A_{(m,Z)}}^0 \ps \\
&= \bra{\psi}\Pi_{B_{(m',\text{Aux})}}^0\Pi_{A_{(m,Z)}}^0 \ps = 0
\end{align}
where to get the second line we used \eqref{eq25}. The proof is analogous for the other cases of \mbox{$k\in\{2m,2m+1\},$} \mbox{$k' \in \{2m',2m'+1\}$} with $m\neq m'$. Hence $P_A^{(k)}P_A^{(k')}\ps = 0$ for all $k \neq k'$, as desired. Condition \eqref{condition:YN-criterion-2} follows easily from condition \eqref{condition:YN-criterion-1}.

\paragraph{Construction of the unitaries}
Finally, to complete the proof of self-testing we just need to construct unitary operators satisfying condition \eqref{condition:YN-criterion-4} of Lemma \ref{YNcriterion}. These operators are obtained in exactly the same way as in \cite{CGS16} as appropriate alternating products of the $X^{\mathscr{u}}_{A,m}$, $Y^{\mathscr{u}}_{A,m}$ for Alice and of the $X^{\mathscr{u}}_{B,m}$, $Y^{\mathscr{u}}_{B,m}$ for Bob. We hence refer the reader to \cite{CGS16} for the last part of the proof. The claimed robustness bound is straightforward, and follows in the same way as for Theorem \ref{thm3}.
\end{proof}

\subsection{Extension to infinite question sets and proof of Theorem \ref{thm2}}
In this section, we complete the proof of Theorem \ref{thm2}.
The ideal correlations of Definition \ref{ideal measurements many questions finite case} are extended to infinite questions sets, with the same answer sets, just as one would expect. Let $\mathcal{X}_{\infty} = \mathbb{N} \times \{Z,X,X'\}$, $\mathcal{Y}_{\infty} = \mathbb{N} \times \{Z,Z',X',\text{Aux}\}$, $\mathcal{A} = \{0,1,2,\perp\}$, $\mathcal{B} = \{0,1,\perp\}$. 
\begin{definition}(Separating correlation, many questions)
The separating correlation in the \textit{many questions} case is the correlation $p^{*}_{\infty}$ on questions sets $\mathcal{X}_{\infty}$, $\mathcal{Y}_{\infty}$ and answer sets $\mathcal{A}, \mathcal{B}$, achieved on the joint state $\ket{\Psi_{\infty}}:= C \cdot \sum_{i=0}^{\infty} \frac{1}{(i+1)^8}\ket{ii}$ (i.e $c_i = C \cdot \frac{1}{(i+1)^8}$, where $C$ is the normalizing constant from section \ref{2.2}), with the ideal measurements of Definition \ref{ideal measurements many questions finite case}, except with $m$ ranging in $\mathbb{N}$.
\end{definition}

Just as in the infinite answers sets case of Section \ref{2.2}, we can also view $p^{*}_{\infty}$ as a limit of ideal correlations from Definition \ref{ideal measurements many questions finite case}, provided we modify the notion of lift, as we do below. 

\begin{proof}[Proof of Theorem \ref{thm2}]

The proof of Theorem \ref{thm2} follows in a very similar fashion to the proof of Theorem \ref{thm1}, making use of the new self-testing result of Theorem \ref{thm4}. So, we highlight just where it differs.

We introduce, first, a modifed notion of \textit{lift}. This time, for any quantum correlation $\{p(a,b|x,y)\}_{x,y}$ on finite question and answer sets $\mathcal{X} = \{0,1,...\frac{d}{2}-1\} \times\{Z,X,X'\}$, $\mathcal{Y} = \{0,\ldots ,\frac{d}{2}-1\} \times\{Z,Z',X', \mbox{Aux}\}$, $\mathcal{A} = \{0,1,2, \perp\}$, $\mathcal{B} = \{0,1, \perp\}$, we define its lift to countably infinite question sets as follows. Given a quantum strategy producing $\{p(a,b|x,y)\}_{x,y}$, the lift is the correlation $\{\hat{p}(a,b,x,y)\}_{x,y}$ on the same answer sets, but question sets $\mathcal{X}_{\infty} = \mathbb{N} \times \{Z,X,X'\}$, $\mathcal{Y}_{\infty} = \mathbb{N} \times \{Z,Z',X',\text{Aux}\}$, obtained by the same quantum strategy, except that when Alice or Bob receive a question in $\mathcal{X}_{\infty} \setminus \mathcal{X}$ and $\mathcal{Y}'_{\infty} \setminus \mathcal{Y}$ respectively, they simply output ``$\perp$" with probability 1.

Now, we define $p^*_N$ to be the ideal correlation from Definition \ref{ideal measurements many questions finite case}, specifically the one obtained on the state $\ket{\Psi_N}= C_N \cdot \sum_{i=0}^{N-1} \frac{1}{(i+1)^8} \ket{ii}$, where $C_N$ is the normalizing constant from section \ref{2.2}, and we let $\mathcal{X}_N, \mathcal{Y}_N$ be the corresponding sets of questions. Then, it is again easy to see (in a similar fashion to Claim \ref{claim1}) that there exists a function $\epsilon(N) = \alpha N^{-16}$, for some constant $\alpha$, such that 
\begin{equation}
|\hat{p}^*_N - 
p^*_{\infty}|_{corr} \leq \epsilon(N)
\end{equation}
Now, we follow through with the same argument and notation. 

Suppose $p \in \mathcal{C}_{q \leq N'}^{\infty, \infty, 4, 3}$ and $|p-p^*_{\infty}|_{corr} = \delta$, for some $\delta > 0$. Then, by a triangle inequality, $|p-\hat{p}^*_N|_{corr} \leq \epsilon(N) + \delta$. Given an $N'$-dimensional quantum strategy achieving $p$, we define $\hat{p}_N$ to be the correlation obtained with the same strategy, except that when Alice or Bob receive a question in $\mathcal{X}_{\infty} \setminus \mathcal{X}_N$ and $\mathcal{Y}_{\infty} \setminus \mathcal{Y}_N$ respectively, they simply output ``$\perp$" with probability 1 (and naturally denote by $p_N$ the correlation on question sets $\mathcal{X}_N$ and $\mathcal{Y}_N$ whose lift is $\hat{p}_N$).
We have,
\begin{align}
\label{eq45}
\forall (x,y) \in \left(\mathcal{X}_{\infty} \setminus \mathcal{X}_N\right) \times \mathcal{Y}_{\infty},  \sum_{(a,b):a \neq \perp} p(a,b|x,y) &= O\big(\epsilon(N)+\delta \big), \\ \label{eq46} \forall (x,y) \in \mathcal{X}_{\infty} \times \left(\mathcal{Y}_{\infty} \setminus \mathcal{Y}_N\right), \sum_{(a,b):b \neq \perp} p(a,b|x,y) & = O\big(\epsilon(N)+\delta \big)
\end{align}
Then, \eqref{eq45} and \eqref{eq46} imply $|p_N  - p^*_N|_{corr} = O\big(\epsilon(N) + \delta\big)$. The rest of the proof mimics part (ii) and (iii) of the proof of Theorem \ref{thm1}.

\end{proof}

\section{Conclusion and open questions}
In conclusion, we have shown separation of finite and infinite-dimensional quantum correlations when one allows for either infinite answer sets or infinite question sets. The proof of the former relies on an extension of the self-testing result from \cite{CGS16} to infinite answer sets. The proof of the latter relies on a novel self-test for any pure bipartite entangled state of local dimension $d$, with question sets of size $O(d)$ and answer sets of size $4$ and $3$ respectively. 

The following are two interesting and related open questions.
\begin{enumerate}[(i)]
\item The major related open question is still, of course, whether the containment $\mathcal{C}_q \subseteq \mathcal{C}_{qs}$ is strict, or the two sets are equal. Proving the conjecture \cite{I3322Pal} that maximal violation of the $I_{3322}$ Bell inequality \cite{I3322Froissart} is attained by an infinite-dimensional quantum state, and not any finite-dimensional one, would imply that  $\mathcal{C}_q\neq \mathcal{C}_{qs}$. On the other hand, it is also possible that correlations violating maximally $I_{3322}$ lie in $\mathcal{C}_{qa}$ but not in $\mathcal{C}_{qs}$ (as Slofstra has shown that $\mathcal{C}_{qs} \neq \mathcal{C}_{qa}$), and that in fact $\mathcal{C}_q = \mathcal{C}_{qs}$. Partial progress in the latter direction would amount to showing, for example, that when one restricts to certain small sizes of questions and answer sets, finite and infinite-dimensional quantum correlations are the same. 

\item Another open question that emerged during this work is whether infinite-dimensional states can be self-tested (with possibly infinite-sized question or answer sets). We suspect that the answer is yes, and that in fact the correlations on infinite question or answer sets that we presented in this work self-test their ideal state. However, the usual self-testing proof techniques don't work in infinite dimensions, because objects like the discrete Fourier transform (and hence the ``swap'' isometry) and the maximally entangled state are not defined. 

\end{enumerate}

\section*{Acknowledgements}
The authors thank Matteo Lostaglio, Martino Lupini, Michal Oszmaniec, William Slofstra and Thomas Vidick for helpful discussions. The authors also thank Thomas Vidick for valuable comments on earlier versions of this paper. A.C.\ is supported by AFOSR YIP award number FA9550-16-1-0495. J.S.\ is supported by NSF CAREER Grant CCF-1553477.

\bibliographystyle{alpha}
\bibliography{references}

\newpage

\appendix
\numberwithin{equation}{section}

\section{Proof of Lemma \ref{YNcriterion}}
\label{YNproof}
We provide a proof of Lemma \ref{YNcriterion}. We will explicitly construct a local isometry $\Phi$ such that $\| \Phi(\ket{\psi})-{\text{extra}}\otimes\ket{\Psi} \| = O(d^{\frac52}\epsilon^{\frac12})$, where $\ket{\Psi}=\sum_{i=0}^{d-1} c_i \ket{ii}$ with $0 < c_i < 1$ for all $i$ and $\sum_{i=0}^{d-1}c_i^2 = 1$. 

\begin{proof}
The first step is to obtain \textit{exactly} orthogonal projections from the $\{P_A^{(k)}\}, \{P_B^{(k)}\}$, which are approximately orthogonal, and only when acting on $\ket{\psi}$, from condition \eqref{condition:YN-criterion-1}. We invoke a slight variation of the \textit{orthogonalization lemma} (Lemma 21) from Kempe and Vidick \cite{KV10}. 
\begin{lem}
\label{projections}
Let $\rho$ be positive semi-definite, living on a finite-dimensional Hilbert space. Let $P_1,..,P_k$ be projections such that 
\begin{equation}
\sum_{i\neq j} \tr(P_iP_jP_i\rho) \leq \epsilon
\end{equation}
for some $0<\epsilon \leq Tr(\rho)$. Then there exist orthogonal projections $Q_1,..,Q_k$ such that
\begin{equation}
\sum\limits_{i=1}^{k} \tr\big((P_i-Q_i)^2\rho\big) =O\big(\epsilon^{\frac12}\big)\tr(\rho)^{\frac12}
\end{equation}

\end{lem}

Note that, importantly, the bound doesn't depend on the dimension of the underlying Hilbert spaces. And it also doesn't depend on the number of projections.

We apply the above Lemma to our projections $\{P_A^{(k)}\}, \{P_B^{(k)}\}$. From condition \eqref{condition:YN-criterion-1}, we have $\sum_{i\neq j}\tr\big(P_A^{(i)}P_A^{(j)}P_A^{(i)}\ket{\psi}\bra{\psi}\big) \leq d(d-1)\epsilon^2$, and similarly for $B$ up to a constant factor, thanks to \eqref{condition:YN-criterion-3} and triangle inequalities. Let $\{\tilde{P}_A^{(k)}\}, \{\tilde{P}_B^{(k)}\}$ be the new sets of orthogonal projections obtained from Lemma \ref{projections}. Then, we have $\sum\limits_{i=0}^{d-1} \tr\big((P_D^{(i)}-\tilde{P}_D^{(i)})^2\ket{\psi}\bra{\psi}\big) =O\big(d\epsilon\big)$, for $D \in \{A,B\}$. This immediately gives $\|(P_D^{(i)}-\tilde{P}_D^{(i)})\ket{\psi} \| =O\big(d^{\frac12}\epsilon^{\frac12}\big)$, for $i=0,..,d-1$, without seeking to optimize the bound further. By application of triangle inequalities, conditions \eqref{condition:YN-criterion-2}, \eqref{condition:YN-criterion-3} and \eqref{condition:YN-criterion-4} become, with the new projections,
\begin{align}
\|(\sum_k \tilde{P}_A^{(k)} - \mathds{1}) \ps \| &= O(d^{\frac{3}{2}}\epsilon^{\frac12}),  \\   
\| (\tilde{P}^{(k)}_{A} - \tilde{P}^{(k)}_{B} )\ket{\psi} \| &= O(d^{\frac{1}{2}}\epsilon^{\frac12}) \;\; \forall k,\\
\| (X^{(k)}_{A}X^{(k)}_{B}\tilde{P}^{(k)}_{A} - \frac{c_k}{c_0} \tilde{P}^{(0)}_{A}) \ket{\psi} \| &= O(d^{\frac{1}{2}}\epsilon^{\frac12}) \;\; \forall k,
\end{align}




Now, define $Z_{A/B} := \sum_{k=0}^{d-1} \omega^k \tilde{P}_{A/B}^{(k)} + \mathds{1} - \sum_k \tilde{P}_{A/B}^{(k)}$.
In particular, $Z_A$ and $Z_B$ are unitary. 

Define the local isometry
\begin{equation}
\Phi := (R_{AA'}\otimes R_{BB'})(\bar{F}_{A'}\otimes\bar{F}_{B'})(S_{AA'}\otimes S_{BB'})(F_{A'}\otimes F_{B'})
\end{equation}
where $F$ is the quantum Fourier transform, $\bar{F}$ is the inverse quantum Fourier transform, $R_{AA'}$ is defined so that $\ket{\phi}_{A} \ket{k}_{A'} \mapsto X^{(k)}_{A}\ket{\phi}_{A} \ket{k}_{A'}\,\,\, \forall \ket{\phi}$, and similarly for $R_{BB'}$, and $S_{AA'}$ is defined so that $\ket{\phi}_{A} \ket{k}_{A'} \mapsto Z^{k}_{A}\ket{\phi}_{A}\ket{k}_{A'} \,\,\, \forall \ket{\phi}$, and similarly for $S_{BB'}$. We compute the action of $\Phi$ on $\ket{\psi}_{AB}\ket{0}_{A'}\ket{0}_{B'}$. For ease of notation with drop the tildes, while still referring to the new orthogonal projections. We write $\ket{\psi} \approx_\epsilon \ket{\psi'}$ to mean $\|\ket{\psi} - \ket{\psi'}\| \leq \epsilon$. 
\begingroup
\allowdisplaybreaks
\begin{align}
\ket{\psi}_{AB}\ket{0}_{A'}\ket{0}_{B'} \stackrel{F_{A'}\otimes F_{B'}}{\longrightarrow}& \frac{1}{d}\sum_{k,k'}\ket{\psi}_{AB}\ket{k}_{A'}\ket{k'}_{B'} \label{A3} \\
\stackrel{S_{AA'}\otimes S_{BB'}}{\longrightarrow}& \frac{1}{d}\sum_{k,k'}\left(\sum_{j}\omega^{j}P^{(j)}_{A} + \mathds{1} - \sum_j P_{A}^{(j)}\right)^{k}\left(\sum_{j'}\omega^{j'}P^{(j')}_{B} + \mathds{1} - \sum_j' P_{B}^{(j')}\right)^{k'}\ket{\psi}_{AB}\ket{k}_{A'}\ket{k'}_{B'} \label{A3}\\
&\approx_{O(d^{\frac52}\epsilon^{\frac12})} \frac{1}{d}\sum_{k,k',j,j'}\omega^{jk}\omega^{j'k'}P^{(j)}_{A}P^{(j')}_{B}\ket{\psi}_{AB}\ket{k}_{A'}\ket{k'}_{B'}\\
&\approx_{O(d^{\frac52}\epsilon^{\frac12})}  \frac{1}{d}\sum_{k,k',j,j'}\omega^{jk}\omega^{j'k'}P^{(j)}_{A}P^{(j')}_{A}\ket{\psi}_{AB}\ket{k}_{A'}\ket{k'}_{B'}\\
&= \frac{1}{d}\sum_{k,k',j}\omega^{j(k+k')}P^{(j)}_{A}\ket{\psi}_{AB}\ket{k}_{A'}\ket{k'}_{B'}\\
\stackrel{\bar{F}_{A'}\otimes \bar{F}_{B'}}{\longrightarrow}&\frac{1}{d^2}\sum_{k,k',j,l,l'}\omega^{j(k+k')}\omega^{-lk}\omega^{-l'k'}P^{(j)}_{A}\ket{\psi}_{AB}\ket{l}_{A'}\ket{l'}_{B'}\\
&=\frac{1}{d^2}\sum_{k,k',j,l,l'}\omega^{k(j-l)}\omega^{k'(j-l')}P^{(j)}_{A}\ket{\psi}_{AB}\ket{l}_{A'}\ket{l'}_{B'}\\
&= \sum_{j}P^{(j)}_{A}\ket{\psi}_{AB}\ket{j}_{A'}\ket{j}_{B'}\label{A12}\\
\stackrel{R_{AA'}\otimes R_{BB'}}{\longrightarrow}& \sum_{j}X^{(j)}_{B}X^{(j)}_{A}P^{(j)}_{A}\ket{\psi}_{AB}\ket{j}_{A'}\ket{j}_{B'}\\
&\approx_{O(d^{\frac32}\epsilon^{\frac12})} \sum_{j}\frac{c_j}{c_0} P^{(0)}_{A}\ket{\psi}_{AB}\ket{j}_{A'}\ket{j}_{B'}\\
&= \frac{1}{c_0}P^{(0)}_{A}\ket{\psi}_{AB} \otimes \sum_{j} c_j \ket{j}_{A'}\ket{j}_{B'}\\
&= \ket{\text{extra}} \otimes \ket{\Psi} \label{A}
\end{align}
\endgroup
All in all, we have constructed a local isometry $\Phi$ such that 
\begin{equation}
\| \Phi(\ket{\psi}) - \ket{\text{extra}} \otimes \ket{\Psi}\| = O(d^{\frac52}\epsilon^{\frac12})
\end{equation}

Note that it is straightforward to check that the whole proof above can be repeated by starting from a mixed joint state, yielding a corresponding version of the Lemma that holds for a general mixed state.

\end{proof}

\section{Proof of Theorem \ref{thm3}}

The proof is mostly a matter of going through the proof for the exact case in \cite{CGS16} and checking that all equalities can be replaced by approximate equalities, making use of triangle inequalities. One then invokes Lemma \ref{YNcriterion}, i.e. the slightly more general and robust version of the self-testing criterion from \cite{Yang13}. We provide a sketch of the proof using the same notation as in \cite{CGS16}. We invite the interested reader to refer to \cite{CGS16}.

First, we clarify some jargon. We say that an equation $\ket{\psi} = \ket{\psi'}$ holds $\epsilon$-approximately, if $\|\ket{\psi} - \ket{\psi'}\| \leq \epsilon$, and we write $\ket{\psi} \approx_\epsilon \ket{\psi'}$. We will go through the proof in section 4 of \cite{CGS16}, pointing out where exact identities are replaced by approximate ones. From here on, we also refer to the equation numbering from section 4 of \cite{CGS16}.

First, notice that $\epsilon$-approximate correlations give us $ \big| \| \Pi^{A_0}_{2m} \ps \|^2 - c_{2m}^2 \big| =O(\epsilon)$. Similarly, $\big| \| \Pi^{A_0}_{2m+1} \ps \|^2 - c_{2m+1}^2 \big| =O(\epsilon)$. With similar other calculations, (10) becomes
\begin{align}
\label{id_norms approx}
\big| \|\mathds{1}_m^{A_i} \ps \|^2 - c_{2m}^2 - c_{2m+1}^2 \big| &= O(\epsilon) \,\,\,\,, i \in \{0,1\} \\ 
\big| \|\mathds{1}_m^{B_i} \ps \|^2 - c_{2m}^2 - c_{2m+1}^2 \big| &= O(\epsilon)  \,\,\,\,, i \in \{0,1\}\,,
\end{align}
which implies, since $a^2-b^2 = (a-b)(a+b)$,
\begin{align}
\label{35}
\big| \|\mathds{1}_m^{A_i} \ps \| - \sqrt{c_{2m}^2 - c_{2m+1}^2} \big| &= O(\epsilon) \,\,\,\,, i \in \{0,1\} \\ 
\big| \|\mathds{1}_m^{B_i} \ps \| - \sqrt{c_{2m}^2 - c_{2m+1}^2} \big| &= O(\epsilon)  \,\,\,\,, i \in \{0,1\}\,.
\end{align}
Then, we have $\average{\psi| \mathds{1}_m^{A_i}\mathds{1}_m^{B_j}| \psi} \geq \|\mathds{1}_m^{A_i} \ps \| \cdot \|\mathds{1}_m^{B_j} \ps \| - O(\epsilon)$. And so (11) becomes 
\begin{equation}
\mathds{1}_m^{A_i} \ps \approx_{O(\sqrt{\epsilon})} \mathds{1}_m^{B_j} \ps \,\,\,\, \forall i,j \in \{0,1\}\,.
\end{equation}
Next, (12) and (13) hold $O(\epsilon)$-approximately, and so, by Lemma \ref{Bamps lemma}, equations (14) and (15) become 
\begin{align}
Z^{\mathscr{u}}_{A,m} \ket{\psi_m} &\approx_{O(\sqrt{\epsilon})} Z^{\mathscr{u}}_{B,m} \ket{\psi_m} \label{eq38}\\
X^{\mathscr{u}}_{A,m} (\mathds{1} - Z^{\mathscr{u}}_{A,m}) \ket{\psi_m} &\approx_{O(\sqrt{\epsilon})} \tan(\theta_m) X^{\mathscr{u}}_{B,m} (\mathds{1} + Z^{\mathscr{u}}_{A,m}) \ket{\psi_m} \label{eq39}
\end{align}
Now, equation (16) holds $O(\sqrt{\epsilon})$-approximately, which implies that 
\begin{equation}
Z^{\mathscr{u}}_{B,m} \ket{\psi_m} \approx_{O(\sqrt{\epsilon})} \tilde{Z}_{B,m} \ps
\end{equation}
Then, equations (17) and (18) hold $O(\sqrt{\epsilon})$-approximately, and analogously does (21).
So, 
\begin{align}
X^{\mathscr{u}}_{A,m} P_A^{(2m+1)} \ps &\approx_{O(\sqrt{\epsilon})} \tan(\theta_m) X^{\mathscr{u}}_{B,m} P_A^{(2m)} \ps  = \frac{c_{2m+1}}{c_{2m}} X^{\mathscr{u}}_{B,m} P_A^{(2m)} \ps  \label{eq41}\\
X^{'\mathscr{u}}_{A,m} P_A^{(2m+2)} \ps &\approx_{O(\sqrt{\epsilon})} \tan(\theta'_m) X^{'\mathscr{u}}_{B,m} P_A^{(2m+1)} \ps = \frac{c_{2m+2}}{c_{2m+1}} X^{'\mathscr{u}}_{B,m} P_A^{(2m+1)} \ps \label{eq42}
\end{align}
Finally, the calculations in (28) require $O(d)$ uses of \eqref{eq41}, \eqref{eq42}. Hence, we get
\begin{equation}
X_A^{(k)} P_A^{(k)}\ps \approx_{O(d\sqrt{\epsilon})} \frac{c_{2m+1}}{c_0}(X_B^{(k)})^{\dagger} P_A^{(0)}\ps
\end{equation}
Applying Lemma \ref{YNcriterion} gives the desired conclusion of Theorem \ref{thm3}. 

\end{document}